\documentclass[preprint,12pt]{elsarticle}

\usepackage{geometry}
\usepackage[latin1]{inputenc}
\usepackage{graphicx}
\usepackage{amssymb}
\usepackage{amsmath}
\usepackage{epstopdf}

\def\bq{\begin{equation}}
\def\eq{\end{equation}}
\def\bqa{\begin{eqnarray}}
\def\eqa{\end{eqnarray}}

\def\ba{\begin{array}}
\def\ea{\end{array}}

\def\pmat{\begin{pmatrix}}
\def\emat{\end{pmatrix}}

\def\calA{{\mathcal A}}

\def\calM{{\mathcal  M}}

\def\bbR{\mathbb R}
\def\bbC{\mathbb C}
\def\bbZ{\mathbb Z}

\newcommand{\Cnp}[2]%
{%
{#2 \choose #1}
}

\def\XXint#1#2#3{{\setbox0=\hbox{$#1{#2#3}{\int}$}
     \vcenter{\hbox{$#2#3$}}\kern-.5\wd0}}

\newtheorem{prop}{Proposition}
\newtheorem{thm}{Theorem}

\newtheorem{definition}{Definition}

\newtheorem{proof}{Proof}

\begin{document}
\begin{frontmatter}
\title{Topology of Bloch Bands from Cauchy Data}

\author{Didier Felbacq,\, Emmanuel Rousseau}
\affiliation{organization={Laboratory Charles Coulomb UMR 5221, University of Montpellier},
            addressline={Place Eugène Bataillon}, 
            city={Montpellier},
            postcode={34095 Cedex 05},
            country={France}}

\begin{abstract}

In a previous work, the topology of inversion-symmetric one-dimensional periodic media was characterized through the pole-zero pattern of an impedance-like function associated with Bloch waves. This construction reproduces the Berry--Zak invariant and provides a criterion for topological interface states.

In the present work, we give a geometric interpretation of this formalism. We show that poles and zeros arise naturally from the action of inversion symmetry on the projectivized space of Cauchy data. The corresponding Dirichlet and Neumann states are identified with the two fixed points of the induced $\bbZ_2$ action on the Riemann sphere.

The key observation is that Bloch eigenvectors are naturally constructed on the universal covering of the Brillouin circle. The topology of the associated Real eigenline bundle is encoded in the action of the deck transformation group on lifted eigenvectors. This action is described by a monodromy sign $\rho\in\{\pm1\}$, determined by the inversion representations carried by the band at the fixed points of the Brillouin zone.

At this point, no characteristic class has to be introduced. The only datum is the sign acquired by a lifted eigenvector after one turn around the Brillouin circle. This sign is the clutching datum of a real line bundle over $S^1$: gluing the two ends of an interval by $+1$ gives the trivial bundle, whereas gluing them by $-1$ gives the Möbius bundle. The first Stiefel--Whitney class is the standard topological invariant which records this distinction. Thus, in the present setting, $w_1$ does not introduce an additional structure; it is the usual topological name for the monodromy sign $\rho$.

We show below that this monodromy is determined by the inversion representations at the two fixed points of the Brillouin circle and coincides with the $\bbZ_2$ pole-zero invariant introduced in \cite{FelbacqAnnalen}. Equivalently, the first Stiefel--Whitney class recovers the $\bbZ_2$ reduction of the pole-zero classification.

\end{abstract}
\end{frontmatter}
\section{Introduction}

Topological band theory is usually formulated in terms of vector bundles over momentum space \cite{topophot,colloquium,simon,zak,thouless}. In periodic systems, the eigenstates of an isolated band define a vector bundle over the Brillouin zone whose characteristic classes provide topological invariants \cite{milnor,nakahara,nash}. In two-dimensional systems, the central invariant is the first Chern class, whose integral gives the Chern number. In one-dimensional systems, the situation is simpler: the complex Bloch bundle over the Brillouin circle is always topologically trivial, and topology can only arise through additional symmetries \cite{kohn,zak,Chan2014,FelbacqAnnalen}.

For one-dimensional inversion-symmetric systems, the relevant invariant is the Berry--Zak phase. Since inversion symmetry quantizes this phase to the values $0$ and $\pi$, one obtains a $\mathbb Z_2$ classification of Bloch bands. In modern language, this classification may be expressed in terms of the first Stiefel--Whitney class of a Real line bundle associated with the band \cite{zirnbauer,atiyah,gomi}.

In a previous work \cite{FelbacqAnnalen}, a different approach was proposed. Instead of introducing a Berry connection, the analysis was based on the Cauchy data of Bloch eigenvectors. If $U=(u,u')$ denotes the Cauchy vector associated with a Bloch mode, one introduces the projective coordinate $\chi=u/u'$ which is a meromorphic function of the wavenumber. The poles and zeros of $\chi$ occur at band edges and form an ordered pattern along the spectrum. It was shown that this pole-zero pattern reproduces the Berry--Zak classification and provides a criterion for the existence of topological interface states.

Although effective, the geometric origin of this construction remained unclear. Why should poles and zeros encode topology? What is the significance of the distinguished vectors $(1,0)$ and $(0,1)$, which play a central role in the theory? How does the pole-zero invariant relate to standard topological notions such as Real bundles, local systems and characteristic classes?

The purpose of the present work is to answer these questions.

Our starting point is the observation that inversion symmetry acts naturally on the space of Cauchy data. After projectivization, the Cauchy-data space becomes the Riemann sphere $\mathbf{CP}^1$ (see fig. \ref{cp1}).

\begin{figure}[h!]
	\begin{center}
		\includegraphics[width=8cm]{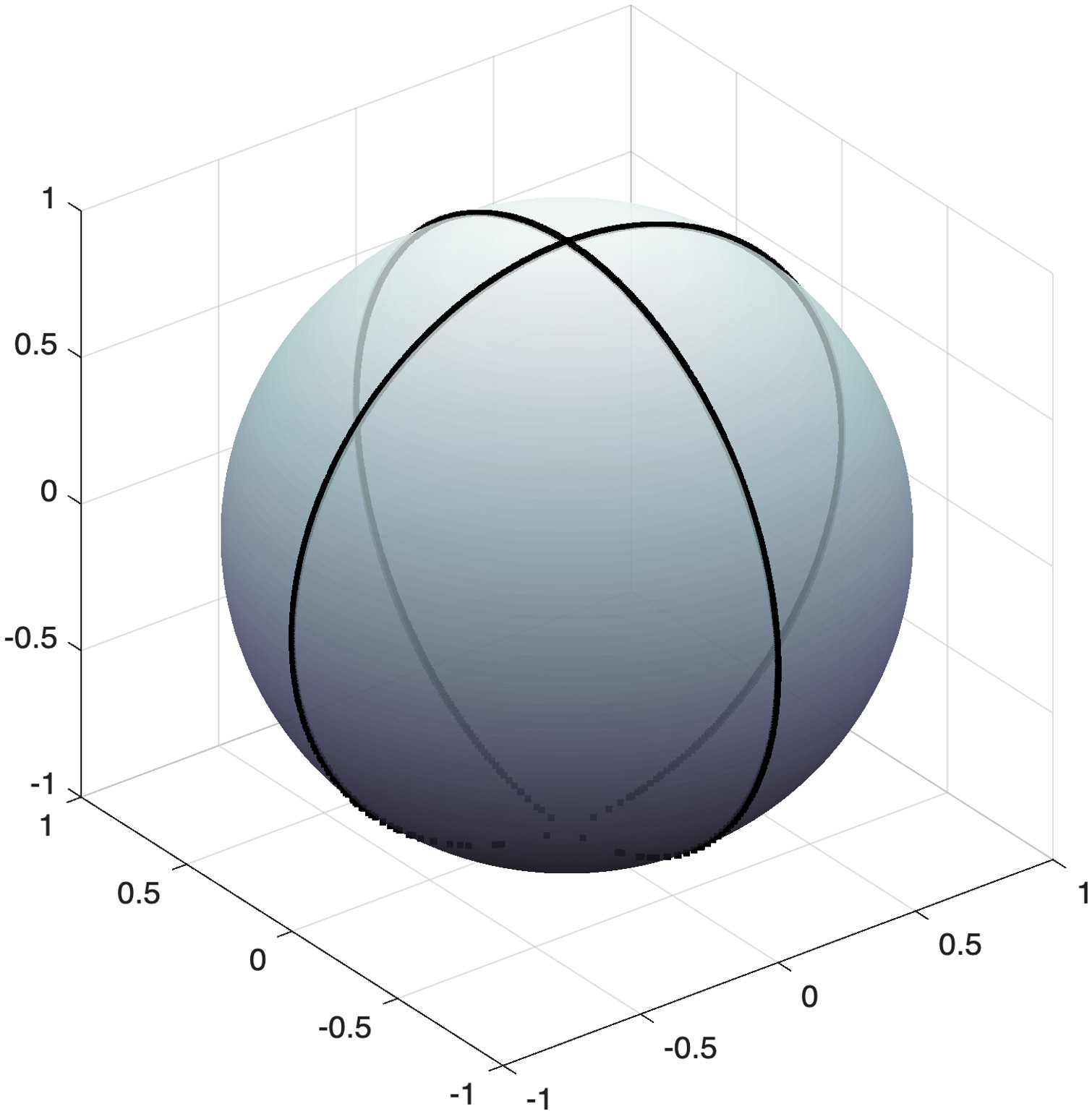}
		\caption{\label{cp1} Trajectory of the Cauchy data on the Riemann sphere obtained numerically by a direct computation of $\chi$. The North and South poles correspond to the Neumann and Dirichlet conditions respectively. }
	\end{center}
\end{figure}

The Dirichlet and Neumann directions appear as the two fixed points of the induced $\mathbb Z_2$ action. Poles and zeros are therefore interpreted as the two irreducible representations of the inversion group acting on projective Cauchy data.
A second key observation is that Bloch eigenvectors are naturally constructed on the universal covering of the Brillouin circle rather than on the circle itself. Let $p:\mathbb R\longrightarrow S^1$ be the universal covering map. The pull-back of the Bloch eigenline bundle over $\mathbb R$ is trivial and therefore admits a global non-vanishing section. The topology of the original bundle is encoded in the action of the deck transformation $q\mapsto q+2\pi$ on such a lifted section.

The crucial point is that the failure of an inversion-equivariant Bloch section to descend from the universal covering naturally defines a monodromy representation :
$$
\rho:\pi_1(S^1)\simeq\bbZ \longrightarrow \{\pm1\}.
$$
This representation determines a rank-one local system over the Brillouin circle. The inversion representations carried by a Bloch band at the two fixed points of the Brillouin zone completely determine this monodromy. Denoting by $s_0,s_\pi\in \{\pm1\}$ the inversion eigenvalues at $q=0$ and $q=\pi$, respectively, one obtains $\rho=s_0s_\pi$.

The main result  (cf. theorem \ref{thm1}) identifies the obstruction to equivariant gluing with a rank-one local system whose monodromy is
$\rho=s_0s_\pi$. This monodromy determines the first Stiefel--Whitney class of the associated Real eigenline bundle.

Once the monodromy representation is known, one may construct the associated local coefficient system \cite{hatcher,kirk}. The resulting twisted cohomology appears naturally from the topology of the band rather than from an \emph{a priori} choice of coefficients. The associated Real eigenline bundle, in the sense of Atiyah \cite{atiyah}, is classified by its first Stiefel--Whitney class. The pole-zero invariant introduced in \cite{FelbacqAnnalen} is then recovered as a concrete analytic realization of this topological structure.

The paper is organized as follows. Section 2 reviews the monodromy description of Bloch waves and introduces the projective space of Cauchy data. Section 3 describes the action of inversion symmetry on projective Cauchy data and revisits the pole-zero classification. Section 4 introduces the equivariant gluing problem and the associated local system. Section 5 identifies the resulting monodromy with the first Stiefel--Whitney class and relates it to the Berry--Zak phase. Section 6 discusses local coefficient systems and twisted cohomology. Section 7 revisits the bulk-edge correspondence from this perspective. Section 8 provides a numerical illustration of the results . Finally, Section~9 presents several perspectives concerning higher-dimensional and non-Hermitian periodic systems.
\section{Bloch bands, projective Cauchy data and inversion symmetry}

We consider a one-dimensional periodic medium described by a scalar Helmholtz equation
$u''+k_0^2V(x)u=0$. The extension to the Schrödinger equation is straightforward.
The potential or material coefficient $V$ is assumed to be bounded and periodic with period one:
$V(x+1)=V(x)$.
We further assume that the medium is inversion-symmetric, so that the origin may be chosen in such a way that
$V(x)=V(-x)$.

The second-order equation may be rewritten as a first-order system for the Cauchy vector
$$
U(x)=
\begin{pmatrix}
u(x)\
u'(x)
\end{pmatrix}^t.
$$
in the form:
$$
\frac{dU}{dx}
=
\begin{pmatrix}
0&1\\
-k_0^2V(x)&0
\end{pmatrix}
U.
$$

Let $\calM(k_0)$ denote the monodromy matrix over one period: $U(1)=\calM(k_0) U(0)$ \cite{felbacq1998,felbacq2003}.
When the equation is real and conservative, the monodromy matrix is unimodular, i.e. $\det(\calM)=1$. Its eigenvalues determine the usual Floquet-Bloch structure. If $|\operatorname{tr}\calM|<2$, the eigenvalues lie on the unit circle and can be written
$z=e^{iq},\, z^{-1}=e^{-iq}$.
The parameter $q$ is the Bloch quasi-momentum, defined modulo $2\pi$. For a fixed isolated band, the Bloch eigenspaces form a complex line bundle over the Brillouin circle $S^1=\bbR/2\pi \bbZ$.
Since every complex line bundle over $S^1$ is topologically trivial, the complex Bloch bundle alone cannot carry a non-trivial Chern-type invariant. The topological information appears only after taking into account the additional inversion symmetry.

The point of view adopted here is that an eigenvector of the monodromy matrix is simultaneously a Cauchy datum. Thus the relevant space is not first the Hilbert space of Bloch modes, but the two-dimensional space of Cauchy data
$\mathcal C=\mathbb C^2$.
Since eigenvectors are defined only up to multiplication by a non-zero scalar, the natural target space is the projectivization
$\mathbf P(\mathcal C)=\mathbf{CP}^1$.
The coordinate $\chi=\frac{u}{u'}$ is the standard affine coordinate on this Riemann sphere.

\begin{figure}[h!]
	\begin{center}
		\includegraphics[width=13cm]{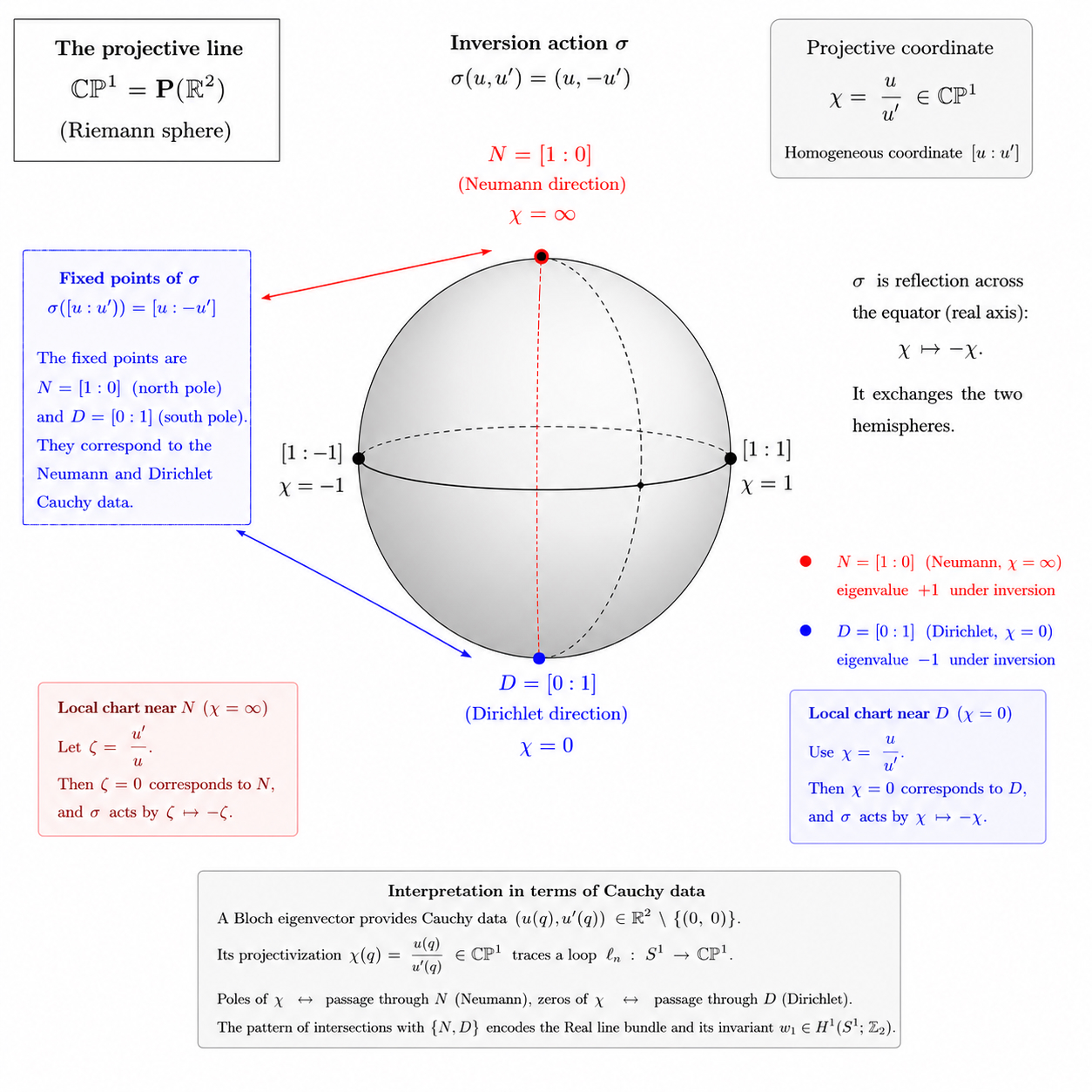}
		\caption{The projective space $\bbC \mathbb{P}^1$ as the space of projectivized Cauchy data. The inversion action $\sigma (u,u')=(u,-u')$ induces a reflection $\chi \to -\chi$ on $\bbC \mathbb{P}^1$ . The fixed points $N=[1:0]$ (Neumann condition, $\chi=\infty$) and $D=[0:1]$ (Dirichlet condition, $\chi=0$) are the poles and the zeros of the projective coordinate $\chi=u/u'$. Their inversion eigenvalues are $+1$ and $-1$ respectively.}
	\end{center}
\end{figure}

\subsection{Inversion symmetry and projective Cauchy data}

We now describe the action of inversion symmetry on the Cauchy-data space. If the origin is chosen at a fixed point of the inversion $x\mapsto -x$, then $u(x)\mapsto u(-x)$, whereas $u'(x)\mapsto -u'(-x)$.
At the fixed point $x=0$, this induces the linear action $\sigma(u,u')=(u,-u')$ on $\mathcal C=\mathbb C^2$, represented by the matrix
$$
\sigma=
\begin{pmatrix}
1&0\\
0&-1
\end{pmatrix}.
$$
The Cauchy-data space decomposes into the two eigenspaces of $\sigma$:
$$
{\mathcal C^+=\{u'=0\},
\,
\mathcal C^-=\{u=0\}.}
$$
The first is the Neumann direction, while the second is the Dirichlet direction. 

The action of $\sigma$ descends to an involution on $\mathbf{CP}^1$: $[u:u']\longmapsto [u:-u']$.
In the coordinate $\chi$ this is simply $\chi\longmapsto -\chi$.
The fixed points of this involution are precisely
$$
\chi=0,
\,
\chi=\infty.
$$
Thus the poles and zeros of $\chi$ are not accidental singularities of a coordinate. They are the two fixed points of the inversion action on the projectivized Cauchy-data space.

\begin{prop}
The Dirichlet and Neumann Cauchy directions are precisely the two projective fixed points of the inversion action on $\mathbf{CP}^1$.
\end{prop}

\begin{proof}
The fixed points of the induced projective action satisfy $[u:u']=[u:-u']$.
This means that there exists $\mu \in\mathbb C^\times$ such that:
$$
u=\mu u,\, -u'=\mu u'.
$$
If $\mu=1$, then $u'=0$, giving the Neumann point $[1:0]$. If $\mu=-1$, then $u=0$, giving the Dirichlet point $[0:1]$. These are the only possibilities.
\end{proof}

\subsection{Pole-zero patterns and inversion representations}

We now recall the pole-zero classification introduced in \cite{FelbacqAnnalen}. For an isolated band, the eigenline of the monodromy matrix defines a map $S^1\longrightarrow \mathbf{CP}^1$.
The affine chart $\chi=u/u'$ defines a meromorphic function $k_0 \to \chi(k_0)$. At band edges, the Bloch multiplier satisfies
$z=\pm1$.
These are precisely the fixed points of the involution $z\longmapsto z^{-1}$ on the Brillouin circle.

At such fixed points, the Bloch eigenspace is invariant under inversion and therefore carries one of the two irreducible representations of $\mathbb Z_2$. These representations are detected by the Cauchy-data direction:
$$
N=[1:0]\quad \Longleftrightarrow \quad s=+1,
$$
and
$$
D=[0:1]\quad \Longleftrightarrow \quad s=-1.
$$
Equivalently, poles correspond to the even representation and zeros correspond to the odd representation.

Let $s_0,s_\pi\in \{\pm1 \}$ denote the inversion eigenvalues of the band at the two fixed points $q=0,\, q=\pi$.
The four possible combinations are
$$
(+,+),\quad (-,-),\quad (+,-),\quad (-,+).
$$
In terms of Pole-Zero patterns, these correspond respectively to
$$
PP,\quad ZZ,\quad PZ,\quad ZP.
$$
The result of \cite{FelbacqAnnalen} may be reformulated as follows.
\begin{prop}
For an inversion-symmetric one-dimensional band, the pole-zero pattern determines the product
$s_0s_\pi$.
The cases $PP$ and $ZZ$ give $s_0s_\pi=+1$,
whereas the cases $PZ$ and $ZP$ give $s_0s_\pi=-1$.
\end{prop}

\begin{table}[h]
\centering
$$
\begin{array}{c|c|c|c|c}
s_0 & s_\pi & \text{pole-zero pattern} & \rho & w_1(L_n)\\
\hline
+1 & +1 & PP & +1 & 0\\
-1 & -1 & ZZ & +1 & 0\\
+1 & -1 & PZ & -1 & 1\\
-1 & +1 & ZP & -1 & 1
\end{array}
$$
\caption{Inversion representations at the two fixed points of the Brillouin circle and the corresponding monodromy of the Real eigenline bundle. Poles correspond to the Neumann direction and to the \(+1\) representation of inversion, while zeros correspond to the Dirichlet direction and to the \(-1\) representation.}
\end{table}
The table summarizes the central point of the construction. The pole-zero pattern determines the inversion eigenvalues $s_0,s_\pi\in \{\pm1 \}$
at the two fixed points of the Brillouin circle. Their product gives the monodromy $\rho=s_0s_\pi$,
and this monodromy determines the first Stiefel--Whitney class of the Real eigenline bundle.
This proposition is the analytic input from the pole-zero formalism. The purpose of the next sections is to identify the geometric meaning of this sign.

\section{Equivariant gluing, local systems and monodromy}

Let $L_n\longrightarrow S^1$ be the complex eigenline bundle associated with an isolated Bloch band. Since $S^1$ is one-dimensional, $L_n$ is trivial as a complex line bundle. However, this statement does not take into account inversion symmetry.

The relevant object is the corresponding Real line bundle in the sense of Atiyah \cite{atiyah}. Recall that a Real bundle is a complex bundle over a space endowed with an involution, together with a compatible lift of this involution to the total space. In the present case, the involution on the Brillouin circle is $q\mapsto -q$, while the lift is induced by inversion symmetry acting on Bloch eigenvectors. The resulting Real structure contains topological information that is invisible at the level of the underlying complex line bundle, i.e., although the underlying complex line bundle is trivial, the resulting Real bundle may be topologically non-trivial.

Let $p:\mathbb R\longrightarrow S^1$ be the universal covering map. Pulling back $L_n$ gives a line bundle $p^*L_n\longrightarrow \mathbb R$. Since $\mathbb R$ is contractible, this lifted bundle is trivial. Hence one may choose a global non-vanishing lifted eigenvector $U(q)\in p^*L_n$.
There is no obstruction at this level.

The obstruction appears when one tries to descend $U$ to the Brillouin circle. The deck transformation group of the covering is
$\pi_1(S^1)\simeq\mathbb Z$,
generated by $q\longmapsto q+2\pi$.
Since the fiber is a real line once the inversion-compatible structure is imposed, the lifted section satisfies
$U(q+2\pi)=\rho U(q)$, with $\rho\in \{\pm1 \}$.

\begin{definition}
The sign $\rho\in \{\pm1\}$ defined by $U(q+2\pi)=\rho U(q)$ is called the monodromy of the Real eigenline bundle associated with the band.
\end{definition}

The monodromy is independent of the particular non-vanishing lifted section. Indeed, replacing $U$ by $fU$, where $f$ is a non-vanishing continuous real function on $\mathbb R$, does not change the sign acquired under the deck transformation.

\begin{prop}
The monodromy sign $\rho$ defines a rank-one local system over the Brillouin circle.
\end{prop}

\begin{proof}
The universal covering $p:\mathbb R\rightarrow S^1$ trivializes the pulled-back eigenline bundle. The deck transformation group is naturally identified with $\pi_1(S^1)\simeq\mathbb Z$, generated by $q\mapsto q+2\pi$. By definition of the monodromy, $U(q+2\pi)=\rho\,U(q)$, with $\rho\in \{\pm1\}$.
Therefore the generator of the fundamental group acts on the fiber by multiplication by $\rho$. This defines a representation $\rho:\pi_1(S^1)\rightarrow \{\pm1 \}$, and hence a rank-one local coefficient system over the Brillouin circle.
\end{proof}

\begin{prop}
For an inversion-symmetric isolated band, the monodromy satisfies $\rho=s_0s_\pi$.
\end{prop}

\begin{proof}
Choose a lifted Bloch section on the fundamental domain
$[0,\pi]$. At the fixed points of the inversion, the section belongs to the eigenspaces
$$
U(0)\in E_{s_0},
\qquad
U(\pi)\in E_{s_\pi}.
$$
To reconstruct the full Brillouin circle, the interval $[\pi,2\pi]$ is obtained from $[0,\pi]$
through the inversion symmetry. The equivariant continuation therefore applies
the inversion representation once at $q=0$ and once at $q=\pi$. Consequently the lifted section acquires the factor
$s_0s_\pi$ after one complete turn around the circle. Hence
$$
U(q+2\pi)=s_0s_\pi\,U(q),
$$
which proves
$$
\rho=s_0s_\pi .
$$
\end{proof}
Thus the pole-zero pattern determines a representation
$$
\rho:\pi_1(S^1)\simeq \mathbb Z \to \{\pm1\}.
$$
The geometric significance of this representation will be clarified in the next section, where it will be interpreted both as the clutching datum of a Real line bundle and as the monodromy of a local coefficient system.

\section{Real bundles, Stiefel--Whitney class and Berry--Zak phase}

The monodromy sign $\rho$ defines a real line bundle over the Brillouin circle. Indeed, the circle may be obtained from an interval by gluing its endpoints. For a real line bundle, the clutching function belongs to $GL(1,\mathbb R)=\mathbb R^\times$.
Up to homotopy, there are only two possibilities: $+1,\, -1$.
These two cases correspond to the trivial real line bundle and the M\"obius line bundle. The first Stiefel--Whitney class classifies real line bundles over $S^1$: $w_1(L_n)\in H^1(S^1;\mathbb Z_2)$.
The two possible monodromies are exactly the two possible values of $w_1$: $\rho=+1\quad\Longleftrightarrow\quad w_1(L_n)=0$,
and $\rho=-1 \quad\Longleftrightarrow\quad w_1(L_n)\neq0$.

Combining this observation with the previous section gives the main result.

\begin{thm}\label{thm1}
Let $L_n\to S^1$ be the Real eigenline bundle associated with an isolated inversion-symmetric Bloch band. Let $s_0,s_\pi\in \{\pm1\}$
be the inversion eigenvalues at the two fixed points of the Brillouin circle. Then $w_1(L_n)=0$ if and only if $s_0s_\pi=+1$, and $w_1(L_n)\neq0$ if and only if $s_0s_\pi=-1$.
Equivalently, the monodromy of the Real eigenline bundle is $\rho=s_0s_\pi$.

The $\mathbb Z_2$ pole-zero invariant introduced in \cite{FelbacqAnnalen} coincides with the first Stiefel--Whitney class of the Real eigenline bundle: $\eta_{\rm PZ}=w_1(L_n)$.
This statement should be understood as a geometric reinterpretation of this invariant.
\end{thm}

\begin{proof}
The pole-zero invariant distinguishes the cases $PP/ZZ$ from the cases $PZ/ZP$. The first pair corresponds to
$s_0s_\pi=+1$, and the second to $s_0s_\pi=-1$.
By the identification of the monodromy with $s_0\,s_\pi$ and the classification of real line bundles over $S^1$, this is precisely the distinction between $w_1(L_n)=0$ and $w_1(L_n)\neq0$.
\end{proof}

This is the main geometric meaning of the pole-zero formalism. Poles and zeros are not merely analytic singularities of an impedance function. They detect the inversion representations at the fixed points of the Brillouin circle. These representations determine the monodromy of the Real eigenline bundle, and this monodromy is measured by the first Stiefel--Whitney class.

\subsection{Local coefficient systems}

The preceding discussion naturally leads to the language of local coefficient systems.
The monodromy:
$$
\rho:\pi_1(S^1)\simeq\bbZ \longrightarrow \{\pm1 \}
$$
defines a rank-one local system. More generally, if $A$ is an abelian group equipped with the action
$a\longmapsto \rho a$, one obtains a local coefficient system $A_\rho$ \cite{hatcher,kirk}. The universal covering
$p:\mathbb R\to S^1$ is trivial, but the deck transformation group acts non-trivially on the coefficient group. Thus the cohomology groups
$H^\bullet(S^1;A_\rho)$ are determined not only by the topology of $S^1$ but also by the monodromy representation $\rho$.

For instance, take $A=\bbZ$ and the non-trivial representation $\rho(1)=-1$.

Consider the standard CW decomposition of the circle with one
0-cell $e^0$ and one 1-cell $e^1$. Its universal covering is the real line $\bbR$, whose cellular
chain complex is:
$$
0 \longrightarrow C_1(\widetilde S^1)
\xrightarrow{\partial}
C_0(\widetilde S^1)
\longrightarrow 0.
$$
Both chain groups are free modules over the group ring $\mathbb Z[\pi_1(S^1)]\simeq\bbZ[t,t^{-1}]$,
and the cellular boundary operator is $\partial = t-1$, where $t$ denotes the generator of $\pi_1(S^1)\simeq\bbZ$.
Given a local coefficient system determined by the representation:
$$
\rho:\pi_1(S^1)\to\{\pm1\},
$$
one tensors this complex with the corresponding $\bbZ[\bbZ]$-module.
The boundary operator becomes:
$$
\partial=\rho(t)-1.
$$
For the non-trivial local system, $\rho(t)=-1$, hence $\partial=-2$. It follows that
$$
H^0(S^1;\bbZ_\rho)=0,
\,
H^1(S^1;\bbZ_\rho)\simeq \bbZ_2.
$$
The equality $H^0(S^1;\bbZ_\rho)=0$ expresses the fact that there is no non-zero element of the coefficient group $\bbZ$
which is invariant under the twisted monodromy $a\mapsto -a$. The group $H^1(S^1;\bbZ_\rho)\simeq\bbZ_2$ is the cohomological signature of this non-trivial monodromy.

It is important to note that the local system is not chosen independently of the physics. It is determined by the band in the following sequence:
$$
\text{band}
\longrightarrow
(s_0,s_\pi)
\longrightarrow
\rho
\longrightarrow
A_\rho.
$$
Thus the twisted cohomology is the natural cohomology theory associated with the Real eigenline bundle of the band.

In the Hermitian case, the same invariant is usually expressed through the Berry--Zak phase. If $\calA$ denotes the Berry connection of a normalized Bloch eigenvector, then the Berry--Zak holonomy is $\operatorname{Hol}(\calA)=\exp\left(-\int_{S^1}\calA\right)$.
In inversion-symmetric one-dimensional systems, the Berry--Zak phase is quantized to $0$ or $\pi$ \cite{kohn,zak}.
The corresponding holonomy therefore takes values in $\{\pm 1\}$. The pole-zero formalism identifies this
quantized holonomy with the monodromy $\rho$.
Thus the monodromy $\rho$ is the Real, connection-free version of the Berry--Zak holonomy.

\begin{prop}
The pole-zero pattern determines the monodromy representation
$$
\rho:\pi_1(S^1)\simeq\mathbb Z\to\{\pm1\}.
$$
Equivalently, it determines the associated rank-one local system.
Different ordered patterns may nevertheless determine the same monodromy.
\end{prop}
\begin{proof}
A pole corresponds to the Neumann fixed direction and therefore to the inversion eigenvalue $+1$, whereas a zero corresponds to the Dirichlet fixed direction and therefore to the inversion eigenvalue $-1$. Hence $PP,\ ZZ \Longrightarrow s_0s_\pi=+1$, while $PZ,\ ZP \Longrightarrow s_0s_\pi=-1$. By the previous proposition: $\rho=s_0s_\pi$. In particular, $PZ$ and $ZP$ determine the same non-trivial local system, although they remain distinct as ordered patterns. Therefore the pole-zero pattern completely determines the monodromy representation and hence the associated local system.
\end{proof}

The local system emerges from the failure of an equivariant Bloch section to descend from the universal covering to the Brillouin circle. The pole-zero pattern determines the inversion representations at the fixed points, which in turn determine the monodromy representation and therefore the associated local system. The importance of the local-system viewpoint is that it
identifies the geometric object underlying the pole-zero formalism. In the original pole-zero description, the invariant
appears as a combinatorial pattern of poles and zeros. The local-system construction shows that this pattern
is in fact the monodromy representation associated with the failure of equivariant gluing of Bloch eigenvectors.

Thus the local system provides the geometric mechanism from which the invariant originates.

\section{A numerical illustration}
In order to give a concrete illustration of the previous constructions, we consider two stratified media whose basic period is made of two homogeneous layers of widths $h_1,h_2$ and permittivities $\varepsilon_1,\,\varepsilon_2$. For each medium, we compute the monodromy matrices $M_1,\,M_2$ at a wavenumber $k_0$, extract the band structures and the functions $\chi_2,\chi_1$ \cite{FelbacqAnnalen,felbacq1998,felbacq2003}. Finally, we consider a finite sample of 10 periods of each structure that we put side by side and compute the transmission spectrum, so as to put in evidence the existence of a boundary mode.
The numerical values characterizing the period structures are the following:
\begin{itemize}
\item Structure 1: $h_1=0.5; h_2=0.5$, \,  $\varepsilon_{1}=3.8,\,\varepsilon_{2}=1$,
\item Structure 2: $h_1=0.3; h_2=0.7$, \,  $\varepsilon_{1}=4.2,\,\varepsilon_{2}=1$.
\end{itemize}
In fig. (\ref{band2}), we give the band structures and the transmission spectrum.
\begin{figure}[h!]
	\begin{center}
		\includegraphics[width=12cm]{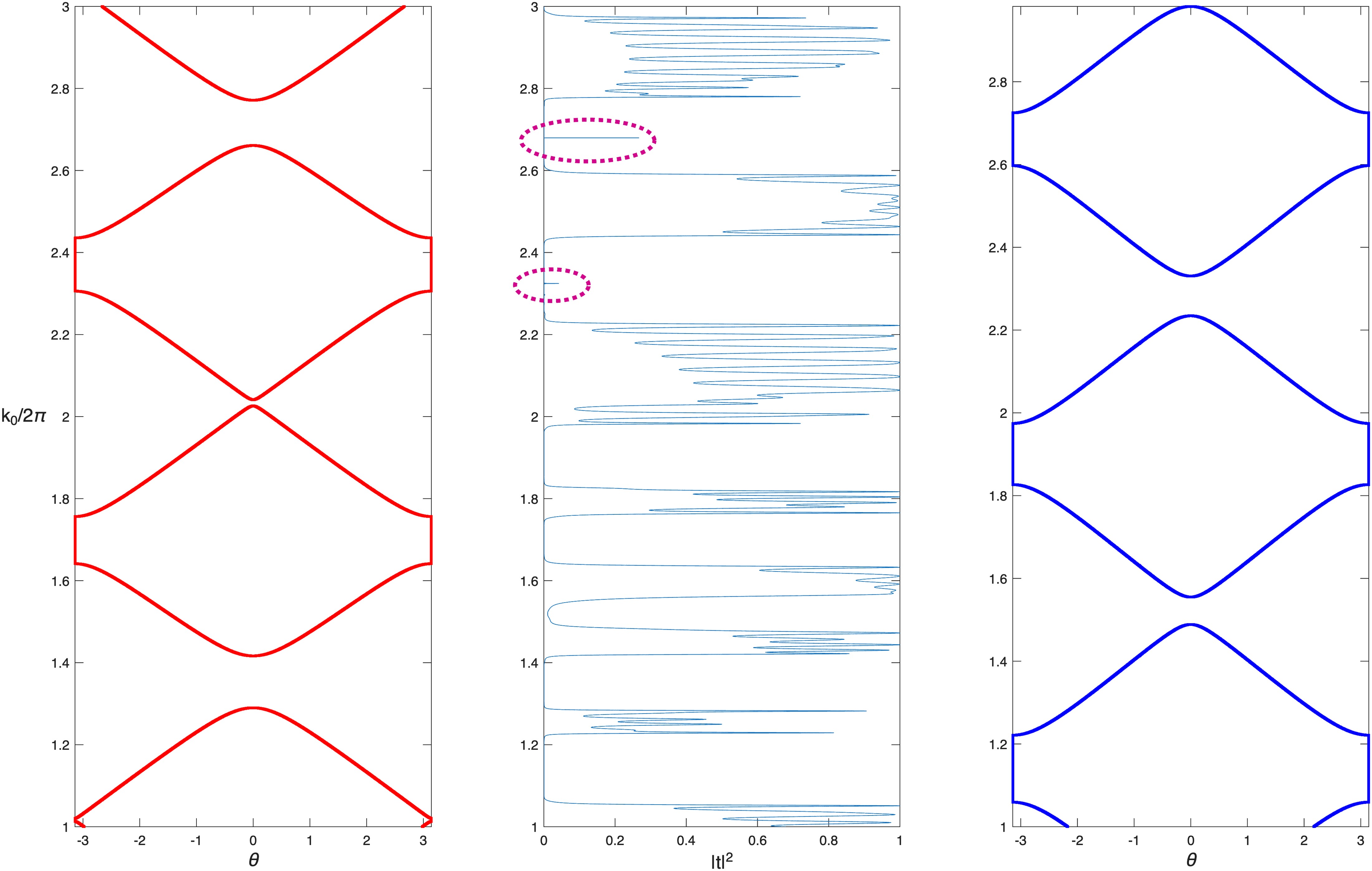}
		\caption{\label{band2} The left and right panels give the band diagram of each structure. Two common band gaps are opened for $k_0 /2\pi$ around $2.3$ and $2.7$. The peaks in the transmission spectrum, surrounded by an ellipsis, correspond to edge modes, i.e. a mode decreasing exponentially in each structure, at these values of the wavenumber. }
	\end{center}
\end{figure}
In fig. (\ref{chim}), we plot the functions $\chi_1$ and $-\chi_2$ as well as the norm commutant of the monodromy matrices, it is null whenever an edge mode appears and $\chi_1=-\chi_2$ \cite{FelbacqAnnalen}. As indicated by the ellipsis, there are two values of $k_0$ for which the functions cross and the monodromy matrices commute: for $k_0/2\pi\simeq 2.3$ and $2.7$. The existence of edge modes inside the gaps is confirmed in fig. \ref{band2}, where two peaks appear in the transmission diagram (indicated by ellipsis). In the common gaps displayed in fig.(\ref{chim}), the two structures exhibit opposite Pole-Zero orderings, namely $ZP$ for $\chi_1$ and $PZ$ for $\chi_2$. Both correspond to the non-trivial value $\rho=-1$, but the opposite ordering gives the impedance crossing criterion of \cite{FelbacqAnnalen}. The observed transmission peaks therefore illustrate the bulk-edge mechanism discussed below.
\begin{figure}[h!]
	\begin{center}
		\includegraphics[width=10cm]{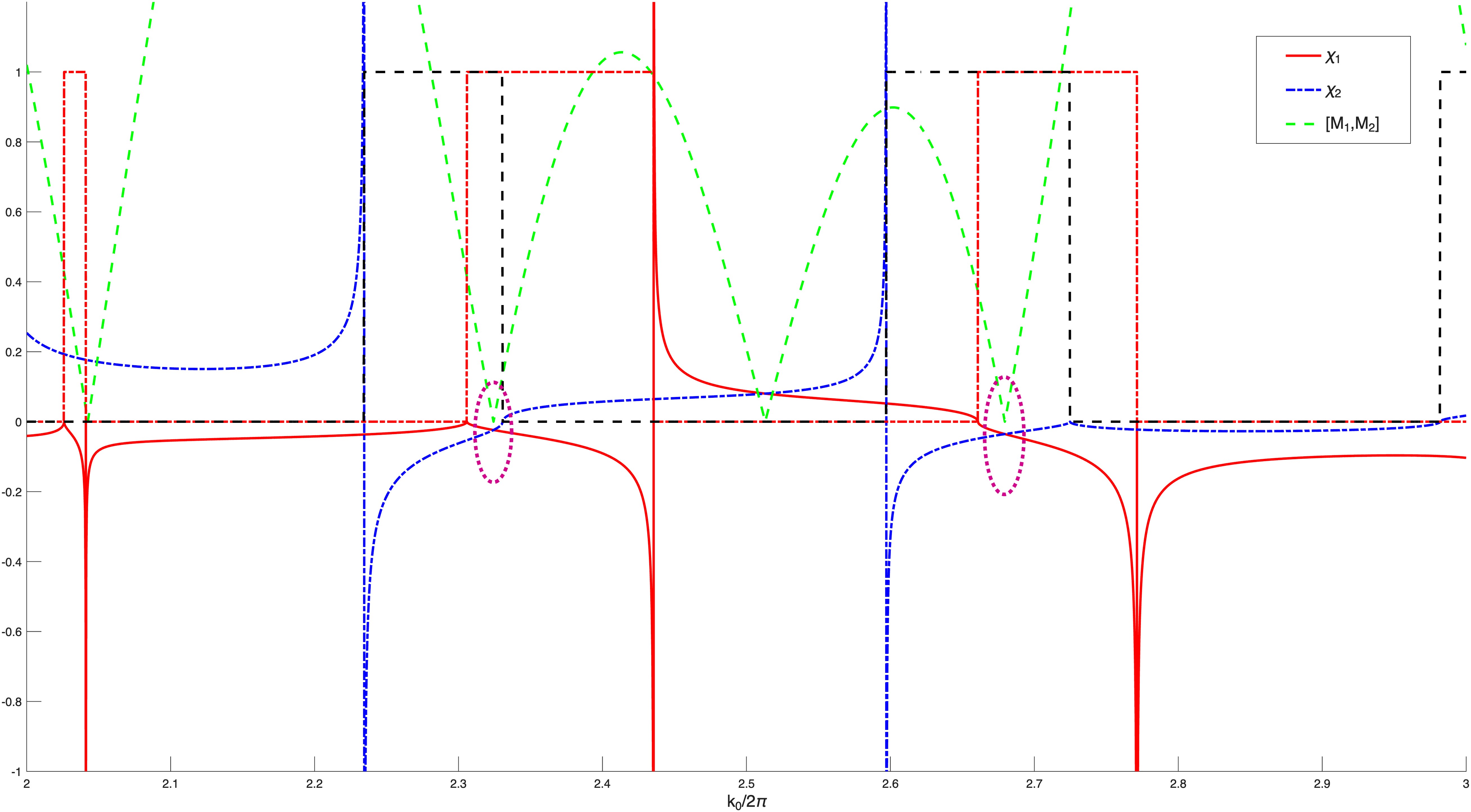}
		\caption{\label{chim} In blue and red: graphs of the functions $\chi_1$ and $-\chi_2$. They cross at values of $k_0$ for which there is an edge mode. Note that for both modes, the Pole-Zero structure is $ZP$ for $\chi_1$ and $PZ$ for $\chi_2$. In green: graph of the norm of the commutant of the monodromy matrices. It is null at an edge mode. The positions of the two edge modes are indicated by the dotted ellipsis. The value 1 of the indicator functions represent the positions of the gaps.}
	\end{center}
\end{figure}

\section{Bulk-edge correspondence and perspectives}
\subsection{The bulk-edge correspondence}

The pole-zero formalism was originally introduced in order to characterize topological interface states in one-dimensional periodic
media. For two semi-infinite inversion-symmetric media joined at an interface, each medium carries a pole-zero pattern in each gap.

The existence of an interface state is governed by the relative position of the pole-zero data of the two media. In the formulation of
\cite{FelbacqAnnalen}, this appears through the crossing condition between the impedance coordinates, namely the crossing of $\chi_1$ and
$-\chi_2$, together with the compatibility condition on the monodromy matrices.

This criterion is more refined than the value of the first Stiefel--Whitney class alone. Indeed, the patterns $PZ$ and $ZP$ both
correspond to the non-trivial monodromy $\rho=-1$ and hence to the same value of $w_1$. Nevertheless, their opposite ordering in the gap
is precisely what produces the impedance crossing and hence the localized interface mode.

Thus the role of the Real-bundle and local-system viewpoint is not to replace the pole-zero bulk-edge criterion by the sole comparison of
Stiefel--Whitney classes. Rather, it explains the geometric origin of the pole-zero data entering that criterion. The edge state reflects a
relative mismatch of projective Cauchy data across the interface, which is detected analytically by the pole-zero ordering and physically by
the crossing of the impedance functions.

\subsection{Perspectives}

The present paper has focused on the one-dimensional Hermitian and inversion-symmetric case, where the geometric structure can be described completely in terms of Real line bundles over the Brillouin circle. Several extensions can be envisioned.

First, the pole-zero construction extends naturally to non-Hermitian systems \cite{FelbacqAnnalen}. In that case the Bloch multipliers and energies become complex, and the natural object is the spectral curve :
$$
\Sigma=\{(E,z):\det(\mathcal M(E)-zI)=0\}.
$$
The eigenline then defines a map:
$$
\Sigma\longrightarrow \mathbf{CP}^1.
$$
From this point of view, poles and zeros are the preimages of the Dirichlet and Neumann fixed points of the inversion action. This suggests that the homology of the spectral curve should play a role analogous to the Brillouin circle in the Hermitian setting.

Second, in higher-dimensional periodic media, the scalar Cauchy datum $(u,u')$ should be replaced by Cauchy data on hypersurfaces or cycles. The analogue of the impedance function is then expected to be a Dirichlet-to-Neumann or Neumann-to-Dirichlet operator. The corresponding projective geometry is no longer the Riemann sphere but an infinite-dimensional space of boundary data. Understanding the action of spatial symmetries on these Cauchy-data spaces may provide a route toward higher-dimensional generalizations of the present construction.

Finally, the local-system viewpoint emphasizes that the topological information carried by a Bloch band is encoded by a monodromy representation of the fundamental group. In the present one-dimensional setting this reduces to the first Stiefel?Whitney class, but the same mechanism may persist in more general situations.
\section{Conclusion}

We have given a geometric interpretation of the pole-zero invariant introduced in \cite{FelbacqAnnalen}. The key observation is that Bloch eigenvectors may be regarded as Cauchy data and hence as points of the projective Cauchy space $\mathbf{CP}^1$.
Inversion symmetry acts on this space, and its two fixed points are precisely the Dirichlet and Neumann directions. Poles and zeros therefore correspond to the two irreducible representations of the inversion group.

The pole-zero pattern determines the inversion eigenvalues at the fixed points of the Brillouin circle. Their product defines the monodromy
$\rho=s_0s_\pi$ of the associated Real eigenline bundle. This monodromy is measured by the first Stiefel--Whitney class
$w_1(L_n)\in H^1(S^1;\mathbb Z_2)$. Thus the $\mathbb Z_2$ pole-zero invariant coincides with $w_1(L_n)$.
The ordered pole-zero pattern contains additional relative information, which enters the interface criterion through the crossing of the
impedance functions.

The construction also explains why local systems appear naturally. Bloch eigenvectors are most naturally constructed on the universal covering of the Brillouin circle, and the deck transformation acts on the lifted section by the monodromy sign $\rho$. This action defines a local coefficient system, whose twisted cohomology records the obstruction to choosing a global equivariant section.

In this way, the pole-zero formalism, the Berry--Zak phase, Real line bundles, local coefficient systems and the first Stiefel--Whitney class are unified within a single geometric framework.

\end{document}